\documentclass[runningheads, envcountsame, a4paper]{llncs}

\usepackage[dvipsnames]{xcolor}
\usepackage{microtype}
\usepackage{amsmath}
\usepackage{amssymb}
\usepackage{stmaryrd}
\usepackage{cite}
\usepackage{paralist}

\newcommand{\bi}{\begin{array}[t]{@{}l@{}}}
\newcommand{\ei}{\end{array}}
\newcommand{\ba}{\begin{array}}
\newcommand{\ea}{\end{array}}
\newcommand{\bda}{\[\ba}
\newcommand{\eda}{\ea\]}
\newcommand{\bp}{\begin{quote}\tt\begin{tabbing}}
\newcommand{\ep}{\end{tabbing}\end{quote}}

\newcommand{\myirule}[2]{{\renewcommand{\arraystretch}{1.2}\ba{c} #1
                      \\ \hline #2 \ea}}

\newcommand{\rlabel}[1]{\mbox{(#1)}}

\newcommand{\conc}{\cdot}

\newcommand{\deriv}[2]{d_{#2}(#1)} 

\newcommand{\nullable}[1]{{\cal N}(#1)}
\newcommand{\isEmpty}[1]{{\Phi}(#1)}
\newcommand{\concPart}[1]{{\cal C}(#1)}
\newcommand{\seqPart}[1]{\mathcal{S}(#1)}

\newcommand{\false}{\mathit{False}}
\newcommand{\true}{\mathit{True}}

\newcommand{\forkEff}[1]{\mathit{Fork(#1)}}

\newcommand{\leftQ}[2]{#1 \backslash #2}

\newcommand\semle\leqq 
\newcommand\semeq\equiv              
\newcommand\syneq=         

\newcommand\DirectDescendant\sqsubset
\newcommand\Descendant\preceq
\newcommand\TrueDescendant\prec

\newcommand\Card{\sharp}
\newcommand\Power{\wp}
\newcommand\Sem[1]{\llbracket#1\rrbracket}

\newcommand\BeforeOrEqual\le

\title{Forkable Regular Expressions}

\author{Martin Sulzmann\inst1  \and Peter Thiemann\inst2}
\institute{
  Faculty of Computer Science and Business Information Systems \\
  Karlsruhe University of Applied Sciences \\
  Moltkestrasse 30, 76133 Karlsruhe, Germany\\
  \email{martin.sulzmann@hs-karlsruhe.de}
  \and 
  Faculty of Engineering, University of Freiburg\\ Georges-K{\"o}hler-Allee
  079, 79110 Freiburg, Germany \\
  \email{thiemann@acm.org}
}

\begin{document}

\maketitle

\begin{abstract}
  We consider  \emph{forkable regular expressions}, which enrich regular expressions with a fork
  operator, to establish a formal basis for static and dynamic analysis of the communication behavior
  of concurrent programs. 
  We define a novel compositional semantics for forkable expressions, establish their fundamental
  properties, and define derivatives for them as a basis for the generation of automata, for
  matching, and for language containment tests.

  Forkable expressions may give rise to non-regular languages, in general, but we identify
  sufficient conditions on expressions that guarantee finiteness of the automata construction via
  derivatives.  
  \keywords{automata and logic, forkable expressions, derivatives}
\end{abstract}

\section{Introduction}
\label{sec:introduction}

Languages like Concurrent ML and Go come with built-in support for fine-grained concurrency, dynamic
thread creation, and channel-based
communication. Analyzing the communication behavior of programs in these languages may be done by
an \emph{effect system}. Such a system computes an abstraction of the sequences of events (i.e., the
communication traces with events like communication actions or synchronizations) that a program may exhibit. 

Effect systems for concurrent programs have been explored by the
Nielsons~\cite{DBLP:conf/popl/NielsonN94}, who proposed to model event traces with ``behaviors'' which
are regular expressions extended with a fork operator that encapsulates the behavior of a newly
created thread. While their work enables the analysis of finiteness properties of the 
communication topology, 
it stops short of providing a semantics of behaviors in terms of effect traces. Subsequent work by the same
authors \cite{DBLP:conf/fase/NielsonAN98,Amt+Nie+Nie:JFP-1997,DBLP:journals/sttt/AmtoftNN98}
concentrates on subtyping and automatic inference of effects. 

We take up the Nielsons' notion of behavior and tackle the problem of defining a compositional semantics
for behaviors in terms of effect traces. Our novel definition yields a semantic basis for the
static and dynamic analysis of concurrent languages with dynamic thread creation and other
communication effects.  We show that, in general, a behavior may give rise to a \emph{non-regular}
trace language. This observation is in line with previous work on concurrent regular
expressions~\cite{Garg:1992:CRE:132251.132252}, flow expressions~\cite{DBLP:journals/tse/Shaw78},
and shuffle expressions~\cite{Gischer:1981:SLP:358746.358767} all of 
which augment regular expressions with (at least) shuffle and
shuffle closure operators. The shuffle closure is also
referred to as the iterated shuffle~\cite{DBLP:journals/tcs/Jantzen85}.

We explore two application areas for forkable expressions. In a run-time verification setting (aka
dynamic analysis), we are interested in matching traces against behaviors,
either at run time or post-mortem. To this end, we extend Brzozowski derivatives to forkable
expressions. Although Brzozowski's construction no longer gives rise to a finite  automaton, in
general, derivatives can still be used to solve instances of the word problem (which is hence decidable).

For static analysis, we are interested in approximation and testing language containment in a
specification. For this use case, we give a decidable criterium that guarantees finiteness of (our
extension of) Brzozowski's automaton construction. This criterium essentially requires a finite
communication topology, that is, it forbids that new communicating threads are created in loops. We conjecture
that this property can be established with the Nielsons' original analysis \cite{DBLP:conf/popl/NielsonN94}.

In summary our contributions are:
\begin{itemize}
\item In Section~\ref{sec:behaviors}, we define a novel trace semantics of behaviors (i.e., regular
  expressions with fork), and establish their fundamental properties.
 \item Section~\ref{sec:derivatives} extends Brzozowski's derivative operation to behaviors.
 \item In Section~\ref{sec:well-behaved-behaviors}, we characterize a
   class of behaviors with regular trace languages. For these behaviors, Brzozowski's construction
   yields finite automata.
\end{itemize}

Related work is discussed in Section~\ref{sec:related-work}.

The online version of this paper contains an appendix with all proofs. \footnote{\url{http://arxiv.org/abs/1510.07293}}

\section{Preliminaries}
\label{sec:preliminaries}

For a set $X$, we write $\Card X$ for the cardinality of $X$ and $\Power (X)$ for its powerset.
If $F : \Power (X) \to \Power (X) $ is a monotone function (i.e., $X\subseteq Y$ implies $F(X)
\subseteq F(Y)$), then we write $\mu F$ for the least fixpoint of this function, which is uniquely
defined due to Tarski's theorem. We write $\mu X.e$ for the fixpoint $\mu (\lambda X.e)$ where $e$ is
a set-valued expression composed of monotone functions assuming that $X$ is a set with the same type
of elements.
(The scope of the $\mu$ operator extends as far to the right as possible.)
We will employ this operator to define
the meaning of forkable Kleene star behaviors~\cite{DBLP:conf/csl/Leiss91}.

Let $\Sigma$ be a finite set, the \emph{alphabet of primitive events}.
We write $\Sigma^*$ for the set of finite words over $\Sigma$ and denote with $v \conc w$ the
concatenation of words $v,w \in \Sigma^*$. 
For languages $L, M \subseteq \Power (\Sigma^*)$, we write $L \conc M = \{ v\conc w \mid v \in L, w
\in M \}$ for the set of all
pairwise concatenations. We write $v\conc M$ as a shorthand for $\{v\} \conc M$. 
The (asynchronous) shuffle operation $v \| w \subseteq \Sigma^*$ on words is the set of all
interleavings of words $v$ and $w$. It is
defined inductively by 
\begin{align*}
  \varepsilon \| w &= \{ w \}  &
  v \| \varepsilon &= \{ v \} &
  xv \| yw & = \{x\} \conc (v \| yw) \cup \{y\} \conc (xv \| w)
\end{align*}
The shuffle operation is lifted to languages by $L \| M = \bigcup \{ v \| w \mid v \in L, w \in M \}$.

We write $\leftQ {L_1} L_2$ to denote the left quotient
of $L_2$ with $L_1$ where 
$\leftQ {L_1} L_2 = \{ w \mid \exists v\in L_1. v \conc w \in L_2  \}$. 
We write $\leftQ x L$
as a shorthand for $\leftQ { \{ x \} } L$.

\section{Behaviors}
\label{sec:behaviors}

\begin{figure}[tp]
  \begin{displaymath}
    \begin{array}[t]{rl}
      L(r) &= L (r, \{  \varepsilon \}) \\
    \end{array}
    \qquad
    \begin{array}[t]{rl}
    L (\phi, K) &= \emptyset \\
    L (\varepsilon, K) &= K \\
    L (x, K) &= \{x\} \conc K
    \end{array}
    \qquad
    \begin{array}[t]{rl}
    L (r+s, K) &= L (r, K) \cup L (s, K) \\
    L (r \conc s, K) &= L (r, L (s, K)) \\
    L (r^*, K) &= \mu X. L (r, X) \cup K \\
    L (\forkEff r, K) &= L (r) \| K
    \end{array}
  \end{displaymath}
  \vspace{-\baselineskip}
  \caption{Trace language of a behavior}
  \label{fig:trace-language}
\end{figure}

Recall that $\Sigma$ is the alphabet of primitive events. Intuitively, a primitive event $x\in\Sigma$ is a globally
visible side effect like sending or receiving a message. A \emph{behavior} is a regular expression
over $\Sigma$ extended with a new fork operator. 
\begin{align*}
  r,s,t &::= \phi \mid \varepsilon \mid  x \mid r+s \mid r \conc s \mid r^* \mid \forkEff{r} \mid ( r )
\end{align*}
As usual, we assume that $\conc$ binds tighter than $+$.

The semantics of a behavior $r$ is going to be a \emph{trace language} $L (r) \subseteq
\Sigma^*$. However, due to the presence of the fork operator, its definition is not a simple
extension of the standard semantics $\Sem{\cdot}$ of a regular expression.
\begin{definition}
  Figure~\ref{fig:trace-language} defines, for a behavior $r$,  
  the trace languages $L (r) \subseteq \Sigma^*$ and 
  $L (r, K)\subseteq \Sigma^*$ with respect to a continuation
  language $K \subseteq  \Sigma^*$.
\end{definition}
By induction on $r$, we can show that the mapping $K \mapsto L (r, K)$
in $\Power (\Sigma^*) \to \Power (\Sigma^*)$ is monotone, so that $L$
is well-defined.
For fork-free behaviors that do not make use of the $\forkEff r$
operator, the trace language is regular and coincides with the standard semantics $\Sem{r}$ of
a regular expression.
\begin{theorem}\label{theorem:fork-free-regular}
  If $r$ is fork-free, then $L(r)$ is regular and $L(r) = \Sem{r} $.
\end{theorem}

It is known that the regular languages are closed under the shuffle operation~\cite{Ginsburg:1966:MTC:1102023}. However, for forkable
expressions the semantics of $\forkEff r$ is defined by \emph{shuffling with the continuation language} so
that the language defined by a behavior need not be regular as the following example shows.
\begin{example}\label{example:behaviors-1}
  Consider the behavior $\forkEff{s}^*$ for a plain regular expression $s$
  where $s$ only consists of the standard regular expression operators.
  Its semantics is the shuffle closure of $\Sem{s}$ as demonstrated by the following calculation
  \begin{align}
    \label{eq:2}
    L (\forkEff{s}^*, \{\varepsilon\}) &
                                         = \mu X. L (\forkEff{s}, X) \cup \{\varepsilon\}
                                         = \mu X. L (s) \| X \cup \{\varepsilon\}\\
    \nonumber{}
                                       &= \{\varepsilon\} \cup L (s) \cup L (s) \| L (s) \cup \dots
  \end{align}
  where we assume that $\|$ binds tighter than $\cup$.
  
  In general, the shuffle closure is not regular~\cite{DBLP:journals/tcs/Jantzen85} as the following concrete instance shows.
  Consider the behavior $r = \forkEff{x\conc y+y\conc x}^*$. By the calculation in \eqref{eq:2}, $L (r)$ is the shuffle
  closure of $\{ x\conc y, y\conc x \}$ which happens to be the context-free language $\{ w \in \{x,y\}^*
  \mid \Card(x, w) = \Card (y, w) \}$ of words that contain the same number of $x$s and $y$s. This
  language is not regular.
\end{example}

\begin{figure}[tp]
  \begin{displaymath}
    \begin{array}{lcl}
  \concPart{\phi} & = & \phi
\\ \concPart{\varepsilon} & = & \varepsilon
\\ \concPart{x} & = & \phi
\\ \concPart{r + s} & = & \concPart{r} + \concPart{s}
\\ \concPart{r \conc s} & = & \concPart{r} \conc \concPart{s}
\\ \concPart{r^*} & = & \concPart{r}^*
\\ \concPart{\forkEff r} & = & \forkEff r
    \end{array}
    \qquad
    \begin{array}{lcl}
\seqPart{\phi} & = & \phi
\\ \seqPart{\varepsilon} & = & \phi
\\ \seqPart{x} & = & x
\\ \seqPart{r + s} & = & \seqPart{r} + \seqPart{s}
\\ \seqPart{r \conc s} & = &
\seqPart{r} \conc s +
\concPart{ r} \conc \seqPart{s}
\\ \seqPart{r^*} & = & \concPart{r}^* \conc \seqPart{r} \conc r^*
\\  \seqPart{\forkEff r} & = & \phi
    \end{array}
  \end{displaymath}
  \caption{Concurrent and sequential part of a behavior}
  \label{fig:seq-conc-part}
\end{figure}

Some of our proofs rely on semantic equivalence and employ identities from 
Kleene algebra~\cite{Kozen:1990:KAC:645720.663687} that
hold for standard regular expressions. Hence, we need to establish that forkable expressions also
form a Kleene algebra.

\begin{definition}[Semantic equality and containment]~\\[-1.5\baselineskip]
  \begin{enumerate}
  \item Behaviors $r$ and $s$ are equal, $r \semeq s$, if $L (r,K) = L (s,K)$, for all $K$.
  \item Behaviors $r$ and $s$ are contained, $r \semle s$, if $L (r,K) \subseteq L (s,K)$, for all $K$.
  \end{enumerate}
\end{definition}
\begin{theorem}
\label{th:behaviors-kleene-algebra}
  The set of forkable expressions with semantic equality and containment is a Kleene algebra. 
\end{theorem}

Each behavior $r$ can be decomposed into a sequential part $\seqPart{r}$ and a concurrent part
$\concPart{r}$, which are defined by induction on $r$ in Figure~\ref{fig:seq-conc-part}. The intuition is that
the sequential part of a behavior describes what must happen next, inevitably, whereas the concurrent
part describes behavior that happens eventually and concurrent to the sequential behavior.
For example, in case of concatenation $r \conc s$, the sequential part
must either start with $\seqPart{r}$, or must end with $\seqPart{s}$.
For Kleene star $r^*$ it is similar, we simply consider the possible unrolling of the underlying
expression $r$.

Our decomposition theorem proves that every behavior is semantically equivalent to the union of its
concurrent part and its sequential part. Its proof requires the Kleene identity $r^* \semeq \varepsilon
+ r \conc r^*$.\footnote{We generally write $r=s$ for \emph{syntactic equality} of expressions and
  use other symbols like $r\semeq s$ for equivalences where some additional reasoning may be involved.}
\begin{theorem}
 \label{le:conc-seq-split}
  For all $r$, $r \semeq \concPart{r} + \seqPart{r}$.
\end{theorem}

The next lemma establishes some algebraic properties of the functions $\concPart{}$ and $\seqPart{}$
that we need in subsequent proofs.
\begin{lemma} 
\label{le:basic-seq-conc-props}
For all $r$:
  \begin{inparaenum}
  \item $\concPart{\concPart r} \syneq {\concPart r}$ (syntactic equality);
  \item $\concPart{\seqPart r} \semeq \phi$;
  \item $\seqPart{\concPart r} \semeq \phi$;
  \item $\seqPart{\seqPart r} \semeq {\seqPart r}$.
  \end{inparaenum}
\end{lemma}
\begin{proof}
  The proof for part~1 is by trivial induction on $r$. See the online version for the remaining
  parts; they are not needed in the rest of this paper.
  \qed
\end{proof}

\begin{lemma}
  \label{le:concpart-nullable}
  For all $r$, $\varepsilon \in L(r)$ iff $\varepsilon \semle \concPart{r}$.
\end{lemma}

\section{Derivatives}
\label{sec:derivatives}

We want to use Brzozowski's derivative operation~\cite{321249}
to translate behaviors to automata and to create algorithms 
for checking language containment and matching.
To this end, we extend derivatives to forkable expressions.
The derivative of $r$ w.r.t.~some symbol $x$,
written $\deriv{r}{x}$, yields the new behavior
after consumption of the leading symbol $x$.
The derivative operation for behaviors is defined by structural induction.
In addition to the regular operators, the derivative needs to
deal with $\forkEff{r}$ expressions and the case of concatenated expressions $r \conc s$ requires
special attention.

\begin{definition}[Derivatives]
The derivative of behavior $r$ w.r.t.~some symbol $x$ is defined
inductively as follows:
\begin{displaymath}
  \begin{array}{lcl}
 \deriv{\phi}{x} & = & \phi
\\
 \deriv{\varepsilon}{x} & = & \phi
\\
 \deriv{y}{x} & = & \left \{
                       \ba{ll}
                          \varepsilon & \mbox{if $x=y$}
                         \\ \phi & \mbox{otherwise}
                       \ea
                     \right.
  \end{array}
  \qquad
  \begin{array}{lcl}
 \deriv{r + s}{x} & = & \deriv{r}{x} + \deriv{s}{x}
\\ 
\color{blue}
    \deriv{r \conc s}{x} & = &
                               \color{blue}
                               \deriv{r}{x} \conc s + \concPart{r} \conc \deriv{s}{x}
    \\
    \deriv{r^*}{x} & = & \deriv{r}{x} \conc r^*
    \\
    \color{blue}
    \deriv{\forkEff{r}}{x} & = &
                                     \color{blue}
                                 \forkEff{\deriv{r}{x}}
  \end{array}
\end{displaymath}
\end{definition}
We just explain the cases that differ from Brzozowski's definition.
The derivative of a fork, $\forkEff{r}$, is simply pushed
down to the underlying expression.
The derivative of $r \conc s$ consists of two components. The first one, $\deriv{r}{x} \conc s$, is
identical to the standard definition: it computes the derivative of $r$ and continues with $s$. 
The second one covers symbols that may reach $s$. In a fork-free regular expression, a symbol in $s$ can
only be consumed if $r$ is nullable, i.e.~$\varepsilon \in L(r)$. For forkable behaviors, a symbol in $s$ can also be consumed if
$r$ exhibits concurrent behavior. 
Hence, we extract the concurrent behavior $\concPart{r}$ and concatenate it with the derivative of  $s$.
The concurrent behavior generalizes nullability in the sense that $\varepsilon \in \concPart{r}$ iff $\varepsilon \in L(r)$.
See Lemma~\ref{le:concpart-nullable}.

Next, we verify that the derivative operation is correct
in the sense that the resulting expression $\deriv{r}{x}$ denotes the left quotient of $r$ by $x$.

\begin{theorem}[Left Quotients]
\label{th:deriv-correctness}
 Let $r$ be a behavior and $x$ be a symbol.
 Then, we have that 
 $L (\deriv{r}{x}) = \leftQ x {L (r)}$.
\end{theorem}
\begin{proof}
  To obtain a viable inductive hypothesis for the proof, 
  we need to
  expand the definition of $L (r) = L (r , \{\varepsilon\})$ 
  and to generalize the statement to an arbitrary
  continuation language $K\subseteq \Sigma^*$. That is, we set out to
  prove, by induction on $r$:
  \begin{equation}\label{eq:1}
    \forall r.\ \forall K.\ 
    L (\deriv{r}{x}, K)  \cup L (\concPart r) \| (\leftQ x K) = \leftQ x {L (r , K)}
  \end{equation}
  The original statement follows from the generalized hypothesis~\eqref{eq:1} by setting
  $K=\{\varepsilon\}$ (recall that $L \| \emptyset = \emptyset$):
  \begin{align*}
    L (\deriv{r}{x})
    & = L (\deriv{r}{x})  \cup L (\concPart r) \| \emptyset \\
    & = 
      L (\deriv{r}{x}, \{\varepsilon\})  \cup L (\concPart r) \| (\leftQ x {\{\varepsilon\}}) \\
    & \stackrel{\eqref{eq:1}}{=} \leftQ x {L (r , \{\varepsilon\})}
        = \leftQ x {L (r)} 
  \end{align*}



  The proof of \eqref{eq:1} proceeds by induction on $r$.
\qed
\end{proof}

Like in the standard regular expression case,
we can conclude (based on the above result) that each behavior
can be represented as a sum of its derivatives.

\begin{theorem}[Representation]
 For any behavior $r$, we have
   $L(r) = (\varepsilon \in L(r) \implies \{ \varepsilon \})
           \cup \bigcup_{x \in \Sigma} x\conc L(\deriv{r}{x})$.
\end{theorem}
Expression $(\varepsilon \in L(r) \implies \{ \varepsilon \})$ denotes $\{ \}$
if $\varepsilon \in L(r)$, otherwise, $\{ \}$.

The representation theorem is the basis for solving the word problem with derivatives. Here, we
extend the derivative operation to words as usual by $\deriv{r}{\varepsilon} = r$ and 
$\deriv{r}{aw} = \deriv{\deriv{r}{a}}{w}$.

\begin{corollary}
  For a behavior $r$ and $w\in\Sigma^*$, $w \in L (r)$ iff $\varepsilon \in L (\deriv{r}{w})$.
\end{corollary}
This corollary implies decidability of the word problem for forkable expressions: the derivative is
computable and the nullability test $\varepsilon \in L (\deriv{r}{w})$ is a syntactic test as for
standard regular expressions.
Full details how to compute all dissimilar derivatives
can be found in the online version.

To construct an automaton from an expression $r$, Brzozowski repeatedly takes the derivative with
respect to all symbols $x\in\Sigma$. We call these derivatives \emph{descendants}.
\begin{definition}[Descendants]
  A \emph{descendant} $s$ of a behavior $r$  is either $r$ itself, a derivative of $r$, or the derivative of a descendant.
  We write $s \DirectDescendant r$, if $s$ is \emph{a direct descendant} of $r$, that is, if $s =
  \deriv{r}{x}$, for some $x$. The ``is descendant of'' relation is the reflexive, transitive closure
  of the direct descendant relation: $s \Descendant r = s \DirectDescendant^* r$.
  The ``is a true descendant of'' relation is the transitive closure
  of the direct descendant relation: $s \TrueDescendant r = s \DirectDescendant^+ r$.
  We define $\deriv{r}{} = \{ s \mid s \Descendant r \}$ as the set of descendants of $r$.
\end{definition}

For standard regular expressions, Brzozowski showed that the set of
descendants of an expression is finite up to \emph{similarity}.
Two expressions are similar if they are equal
modulo associativity, commutativity, and idempotence.
This result no longer holds in our setting.

In the following, we write 
$r \stackrel{x}{\longrightarrow}  s$ if $s = \deriv{r}{x}$.
Subterms on which the derivation operation is applied are underlined.
\begin{example}
\label{ex:infinite-desc}
Let $r = (\forkEff{x \conc y})^*$ and take the derivative by $x$ repeatedly.
\bda{ll}
 & (\forkEff{x \conc y})^*
\\ \stackrel{x}{\longrightarrow} & \underline{\forkEff{y} \conc r}
\\ \stackrel{x}{\longrightarrow} & \forkEff{\phi} \conc r + \underline{\forkEff{y}} \conc \underline{\forkEff{y} \conc r}
\\ \stackrel{x}{\longrightarrow} & \dots + \forkEff{\phi} \conc \forkEff{y} \conc r + \forkEff{y} \conc (\forkEff{\phi} \conc r + \forkEff{y} \conc \forkEff{y} \conc r)
\\ \stackrel{x}{\longrightarrow} & \dots
\eda
Here we omit parentheses (assuming associativity) and 
apply equivalences such as $\concPart{\concPart{r}} \syneq \concPart{r}$ (Lemma~\ref{le:basic-seq-conc-props}).
Clearly, we obtain an increasing sequence of behaviors of the form 
$\forkEff{y} \conc ... \conc \forkEff{y} \conc r$.
Hence, the set of descendants of $r$ is infinite
even if we consider behaviors equal modulo 
associativity, commutativity, and idempotence of alternatives.
\end{example}

This observation is no surprise, given that behaviors may give rise to non-regular
languages (cf.\ Example~\ref{example:behaviors-1}).
In general, there is no hope to retain Brzozowski's result, but it turns out that we can find a
well-behavedness condition for behaviors that is sufficient to retain finiteness of descendants. 

\section{Well-Behaved Behaviors}
\label{sec:well-behaved-behaviors}

In this section, we develop a criterion to guarantee that a forkable expression only gives rise to a
finite set of dissimilar descendants. To start with, 
we adapt Brzozowski's notion of similarity to our setting.
In addition to associativity, commutativity, and idempotence
we introduce simplification rules that implement further Kleene identities and that deal
with forks. 

\begin{figure}[tp]
\bda{cc}
\rlabel{Refl, Trans, Sym, Comp} &
r \eqsim r
\qquad
\myirule{r \eqsim s \quad s \eqsim t}
        {r \eqsim t}
        \qquad
\myirule{s \eqsim t}{t \eqsim s}
\qquad
\myirule{s \eqsim t}{E[s] \eqsim E[t]}
\\[3ex]
\rlabel{Assoc, Comm} &
r + (s + t) \eqsim (r + s) + t
\qquad
r + s \eqsim s + r
\\
\\
\rlabel{Idem, Unit} & r + r \eqsim r
\qquad
r + \phi \eqsim r
\qquad
\phi + r \eqsim r
\\
\\
\rlabel{Empty Word} &
\varepsilon \conc r \eqsim r
\qquad r \conc \varepsilon \eqsim r
\qquad \varepsilon^* \eqsim \varepsilon
\qquad \color{blue} \forkEff{\varepsilon} \eqsim \varepsilon
\\
\\
\rlabel{Empty Language} &
\phi \conc r \eqsim \phi
\qquad r \conc \phi \eqsim \phi
\qquad \phi^* \eqsim \varepsilon
\qquad \color{blue} \forkEff{\phi} \eqsim \phi
\\
\\
\rlabel{Regular Contexts} &
E ::= [] \mid E^* \mid E\conc s \mid r \conc E \mid E+ s \mid r + E \mid \forkEff{E}
\eda
  \caption{Rules and axioms for similarity}
  \label{fig:similarity}
\end{figure}

\begin{definition}[Similarity]
  Behaviors $r$ and $s$ are \emph{similar},
  if $r \eqsim s$ is derivable using the rules and axioms in
  Figure~\ref{fig:similarity}. 
\end{definition}
The compatibility rule \rlabel{Comp} uses regular contexts $E$, which 
are regular expressions with a single hole $[]$. In the rule, we write $E[t]$ to denote
the expression  with the
hole replaced by $t$.

We establish some basic results for similar behaviors, all with straightforward inductive proofs:
Similarity implies semantic equivalence, it is complete for recognizing $\varepsilon$ and
$\phi$, and it is compatible with derivatives and extraction of concurrent parts.

\begin{lemma}\label{le:similarity-is-correct}
  If $r \eqsim s$, then $r \semeq s$.
\end{lemma}

\begin{lemma}
\label{le:eqsim-eps-phi}
\begin{inparaenum}
\item If $L(r) = \{ \varepsilon \}$ then $r \eqsim \varepsilon$.
\item If $L(r) = \{ \}$ then $r \eqsim \phi$.
\end{inparaenum}
\end{lemma}

\begin{lemma}
 \label{le:sim-deriv-conc}
  \begin{enumerate}
  \item If $r \eqsim r'$, then $\deriv{r}{x} \eqsim \deriv{r'}{x}$, for all $x\in\Sigma$.
  \item If $r \eqsim r'$, then $\concPart{r} \eqsim \concPart{r'}$.
  \end{enumerate}
\end{lemma}


Similarity is an equivalence relation. 
We write $[s] = \{ t \mid t \eqsim s \}$ to denote the equivalence class of 
all expressions similar to $s$. If $R$ is a set of
behaviors, we write $R/{\eqsim} = \{ [r] \mid r \in R \}$ for the set of equivalence classes of
elements of $R$.

To identify the set of well-behaved behaviors, we need
to characterize the set of dissimilar descendants.
First, we establish that  each composition of derivatives and applications of $\concPart{}$
that finishes in some $\concPart{r}$ may be compressed to the composition of the derivatives applied to
the remaining $\concPart{r}$.

\begin{lemma}
\label{le:c-d-c=d-c}
For a behavior $r$ and symbol $x$,  $\concPart{\deriv{\concPart{r}}{x}} \syneq \deriv{\concPart{r}}{x}$, syntactically.
\end{lemma}
The above result makes it easier to classify
the forms of dissimilar descendants.

The Kleene star case is clearly highly relevant.
The following statement confirms the observation
in Example~\ref{ex:infinite-desc}.

\begin{lemma}
\label{le:derivatives-of-r*}
  For $w \in \Sigma^+$, 
  $
  \deriv{r^*}{w} \eqsim \deriv{r}{w} \conc r^* + t
  $
  where $t$ is a possibly empty sum of terms of the form 
  \quad$
  s_1 \conc \dots \conc s_n \conc r' \conc r^*
  $\quad
  where $r' \TrueDescendant r$, $n\ge1$,
  and for each $s_i$, $s_i \Descendant \concPart{s}$ for some descendant $s \TrueDescendant r$.
\end{lemma}
\begin{proof}
  Induction on $w$.

  \textbf{Case }$x$: $\deriv{r^*}{x} = \deriv{r}{x} \conc r^* \eqsim \deriv{r}{x} \conc r^* + \phi$.

  \textbf{Case }$wx$: $\deriv{r^*}{wx} = \deriv{\deriv{r^*}{w}}{x}$. By induction for $w$, 
  $\deriv{r^*}{w} \eqsim \deriv{r}{w} \conc r^* + t$ where each summand of $t$ has the form $s_1 \conc
  \dots \conc s_n \conc r'  \conc r^*$. First,  observe that $\deriv{\deriv{r}{w} \conc r^* + t}{x}
  = \deriv{r}{wx} \conc r^* + \concPart{\deriv{r}{w}} \conc \deriv{r}{x} \conc r^* + \deriv{t}{x}$.
  We show by auxiliary
  induction on $n$ that the derivative of $t$ is a sum of terms of the desired form.

  \textbf{Case }$0$:~%
  \begin{minipage}[t]{0.7\linewidth}
    \vspace{-1.8\baselineskip}
    \begin{align*}
      \deriv{r' \conc r^*}{x} 
      &= \deriv{r'}{x}\conc r^* + \concPart{r'} \conc \deriv{r^*}{x} \\
      &= \deriv{r'}{x}\conc r^* + \concPart{r'} \conc \deriv{r}{x} \conc {r^*} 
    \end{align*}
    which has the desired format.
  \end{minipage}
  \medskip{}

  \textbf{Case }$n>0$:\quad%
  \begin{minipage}[t]{0.7\linewidth}
    \vspace{-1.8\baselineskip}
    \begin{align*}
      & \deriv{s_1 \conc \dots \conc s_n \conc r' \conc r^*}{x} \\
      &= \deriv{s_1}{x} \conc s_2 \conc \dots \conc s_n \conc r' \conc r^*
        + \concPart{s_1} \conc \deriv{s_2 \conc \dots \conc s_n \conc r' \conc r^*}{x}
    \end{align*}
  \end{minipage}

  The first summand has the desired form. By induction (on $n$), each summand of $\deriv{s_2 \conc
    \dots \conc s_n \conc r' \conc r^*}{x}$ has the desired form and multiplying with
  $\concPart{s_1}$ from the left retains this form: By Lemma~\ref{le:c-d-c=d-c},
  $\concPart{s_1}$ is still a descendant of $\concPart{s}$ for some descendant $s$ of $r$.
\qed
\end{proof}

To obtain finiteness it appears that a sufficient condition
is to ensure that the subterms $s_i$ are trivial (either $\varepsilon$
or $\phi$). Via similarity, the explosion of terms derived from
Kleene star can then be avoided. To verify this claim
we also characterize the descendants
of concatenated behaviors.

\begin{lemma}
\label{le:form-derivative-concatenation}
  For $w\in\Sigma^+$, $\deriv{r\conc s}{w}$ has the form
  \begin{align*}
    \deriv{r\conc s}{w} &\eqsim \deriv{r}{w} \conc s + \concPart{r} \conc \deriv{s}{w} + t
  \end{align*}
  where $t$ is a sum of terms of the form $r' \conc s'$ where $s'$ is a descendant of $s$ and $r'$ is a descendant of
  $\concPart{r''}$ and $r''$ is a descendant of $r$.
\end{lemma}
\begin{proof}
  By induction on $w$.

  \textbf{Case }$x$: Immediate from the definition of $\deriv{r\conc s}{x}$ with $t=\phi$.

  \textbf{Case }$wx$: By induction
  \begin{align*}
    \deriv{\deriv{r\conc s}{w}}{x}
    &\eqsim \deriv{\deriv{r}{w} \conc s + \concPart{r} \conc \deriv{s}{w} + t}{x} \\
    &\eqsim \underline{\deriv{r}{wx} \conc s} +
    \concPart{\deriv{r}{w}} \conc \deriv{s}{x}
    \\&
    + \deriv{\concPart{r}}{x} \conc \deriv{s}{w} +
    \underline{\concPart{r} \conc \deriv{s}{wx}}
    \\&
    + \deriv{t}{x}
  \end{align*}
  The underlined summands have the expected forms. The newly created summands have a form
  corresponding to $r'\conc s'$. It remains to observe that the derivative of a summand in $t$ has the
  expected form by Lemma~\ref{le:c-d-c=d-c} and Lemma~\ref{le:basic-seq-conc-props}.
  \begin{align*}
    \deriv{r' \conc s'}{x } &= \deriv{r'}{x} \conc s' + \concPart{r'} \conc \deriv{s'}{x}
  \end{align*}
  \vspace{-2.8\baselineskip}

  ~\qed 
\end{proof}

\begin{definition}[Well-behaved Behaviors]
  A behavior $t$ is \emph{well-behaved}
  if all subterms of the form $r^*$ have the
  property that $\concPart{\deriv{r}{w}} \semle \varepsilon$, for all $w \in \Sigma^*$.
\end{definition}

The intuition for this definition is simple: Well-behaved behaviors do not fork processes with
non-trivial communication behavior in a loop (i.e., under a star).
Indeed, we have a simple decidable sufficient condition for well-behavedness.
\begin{lemma}
  If $r \eqsim r'$ and $r'$ is fork-free, then $\concPart{\deriv{r}{w}} \semle \varepsilon$, for all $w
  \in \Sigma^*$.
\end{lemma}
Thus, a behavior is also well-behaved if, for all subterms of the form $r^*$, $r$ is similar to a
fork-free expression.

Recall that $\deriv{r}{}$ is the set of all descendants of $r$ and $\deriv{r}{} / (\eqsim)$
denotes the set of equivalence classes of descendants of $r$. If we pick a representative 
from each of these equivalence classes, we obtain the dissimilar descendants of $r$.
In a practical implementation, we may want to compute the 
\emph{canonical} representative of each equivalence class.
See the online version for further details. 
For the purpose of this paper, an \emph{arbitrary} representative
is sufficient.

\begin{definition}[Dissimilar Descendants]
  We define the set of \emph{dissimilar descendants of $r$},  $\deriv{r}{\eqsim}$, as a complete set of
  arbitrarily chosen representative behaviors for the equivalence classes $\deriv{r}{} / (\eqsim)$. 
\end{definition}

We extend the function $\concPart{}$ on behaviors pointwise to sets of behaviors and
relation $\eqsim$ to sets of behaviors by
\begin{displaymath}
  R \eqsim S \text{ iff }(\forall r \in R. \exists s \in S. r \eqsim s) \wedge (\forall s \in S. \exists r \in R. r \eqsim s)
\end{displaymath}

\begin{lemma}
\label{le:drop-conc-deriv-conc-deriv}
For any behavior $r$, 
$\concPart{\deriv{\concPart{\deriv{r}{\eqsim}}}{\eqsim}} \eqsim \deriv{\concPart{\deriv{r}{\eqsim}}}{\eqsim}$.
\end{lemma}
\begin{proof}
Follows from Lemma~\ref{le:c-d-c=d-c} and Lemma~\ref{le:sim-deriv-conc}.
\qed
\end{proof}

\begin{lemma}
\label{le:deriv-rep-eqsim}
For any behavior $r$, $\deriv{r}{} \eqsim \deriv{r}{\eqsim}$.
\end{lemma}
\begin{proof}
We need to verify that 
for each $t_1 \in \deriv{r}{}$ there  exists $t_2 \in \deriv{r}{\eqsim}$
such that $t_1 \eqsim t_2$. 
We prove this property
by induction on the number of derivative steps.

\textbf{Case $w=\varepsilon$:} Then, $t_1 = r$. Clearly, there exists
$t_2 \in \deriv{r}{\eqsim}$ such that $t_1 \eqsim t_2$.

\textbf{Case $w = x \conc w'$:} Then, $t_1 = \deriv{\deriv{r}{w'}}{x}$.
By the IH, $\deriv{r}{w'} \eqsim t_2$ where $t_2 \in \deriv{r}{\eqsim}$.
By Lemma~\ref{le:sim-deriv-conc}, 
  $t_1 \equiv \deriv{t_2}{x}$ where $\deriv{t_2}{x} \eqsim t_3$
for some $t_3 \in \deriv{r}{\eqsim}$. Thus, we are done.
\qed
\end{proof}

The next result can be verified via similar reasoning.

\begin{lemma}
\label{le:deriv-finite-fp}
For any behavior $r$, $\deriv{\deriv{r}{\eqsim}}{} \eqsim \deriv{r}{\eqsim}$.
\end{lemma}

\begin{theorem}[Finiteness of Well-Behaved Dissimilar Descendants]~\\
  \label{th:finiteness-well-behaved}
 Let $t$ be a well-behaved behavior.
 Then, $\Card{\deriv{t}{\eqsim}} < \infty$.
\end{theorem}
\begin{proof}
 We need to generalize the statement to obtain the result:
 If $t$ is well-behaved then
 $\Card d_t^i < \infty$ for all $i \geq 0$
   where $d_t^0 = \deriv{t}{\eqsim}$ and $d_t^{n+1} = \deriv{\concPart{d_t^n}}{\eqsim}$.

 Based on Lemmas~\ref{le:deriv-finite-fp} and \ref{le:drop-conc-deriv-conc-deriv} 
we find that $d_t^{n+1} = d_t^n$ for $n \geq 1$.
That is, in the induction step it is sufficient to establish
that $d_t^0$ and $d_t^1$ are finite.

We proceed by induction on $t$. For brevity,
we only consider the case of concatenation.

\noindent \textbf{Case $r \conc s$:} By the IH, $d_r^i$ and $d_s^i$ are finite for any $i \geq 0$.
We first show that $d_{r \conc s}^0$ is finite.
 \begin{enumerate}
    \item By Lemma~\ref{le:form-derivative-concatenation}, the elements of $\deriv{r \conc s}{}$ are
      drawn from the set
       $$
        \deriv{r}{} \conc s + \concPart{r} \conc \deriv{s}{} + \sum \deriv{\concPart{\deriv{r}{}}}{} \conc \deriv{s}{}
       $$
    \item By Lemma~\ref{le:deriv-rep-eqsim}
          the above is similar to         
       $$
        \deriv{r}{\eqsim} \conc s + \concPart{r} \conc \deriv{s}{\eqsim} + \sum \deriv{\concPart{\deriv{r}{\eqsim}}}{\eqsim} \conc \deriv{s}{\eqsim}
       $$ 
    \item Immediately, we can conclude that $d_{r \conc s}^0$ is finite.
 \end{enumerate}
  Next, we consider $d_{r \conc s}^1$.
 \begin{enumerate}
   \item From above, the (dissimilar) descendants of $r \conc s$ are drawn from
       $$
        \underbrace{\deriv{r}{} \conc s}_{t_1} + 
        \underbrace{\concPart{r} \conc \deriv{s}{}}_{t_2} + 
       \sum \underbrace{\deriv{\concPart{\deriv{r}{}}}{} \conc \deriv{s}{\eqsim}}_{t_3}
       $$ 
      For each $t_i$ we will show that $d_{t_i}^1$ is finite and thus follows
      the desired result.
   \item By Lemma~\ref{le:form-derivative-concatenation}, descendants of $\concPart{t_1}$ are of the form
    $$
   \deriv{\concPart{\deriv{r}{}}}{} \conc \concPart{s}
   + \concPart{\deriv{r}{}} \conc \deriv{\concPart{s}}{}
   + \sum \deriv{\concPart{\deriv{\concPart{\deriv{r}{}}}{}}}{} \conc \deriv{\concPart{s}}{}
    $$
 \item As we know by the IH, $d_r^i$ and $d_s^i$ are finite for any $i \geq 0$.
       Hence, via similar reasoning as above we can conclude that the above
      is similar to expressions of the form
     \bda{c}
         d_r^1 \conc \concPart{s} + \concPart{d_r^0} \conc d_s^1
       + \sum d_r^2 \conc d_s^1
     \eda
      As all sub-components are finite,
      we conclude that $d_{t_1}^1$ is finite.
  \item We consider the descendants of $\concPart{t_2}$ which are of the following form
  \bda{c}
   \deriv{\concPart{r}}{} \conc \concPart{\deriv{s}{}}
   + \concPart{r} \conc \deriv{\concPart{\deriv{s}{}}}{}
   + \sum \deriv{\concPart{\deriv{\concPart{r}}{}}}{} \conc \deriv{\concPart{\deriv{s}{}}}{}
 \eda
   \item The above is similar to
   \bda{c}
       d_r^1 \conc \concPart{d_s^0}
     + \concPart{r} \conc d_r^1
     + \sum d_r^2 \conc d_s^1
   \eda
    Thus, we find that $d_{t_2}^1$ is finite.
   \item Finally, we observe that shape of descendants of $\concPart{t_3}$
    \bda{l}
        \deriv{\concPart{\deriv{\concPart{\deriv{r}{}}}{}}}{} \conc \concPart{\deriv{s}{}}
  + \concPart{\deriv{\concPart{\deriv{r}{}}}{}} \conc \deriv{\concPart{\deriv{s}{}}}{}
 \\  + \sum \deriv{\concPart{\deriv{\concPart{\deriv{\concPart{\deriv{r}{}}}{}}}{}}}{}
         \conc \deriv{\concPart{\deriv{s}{}}}{}
    \eda
    \item The above is similar to
     \bda{c}
          d_r^2 \conc \concPart{d_s^0}
       +  \concPart{d_r^1} \conc d_s^1
       + \sum d_r^3 \conc d_s^1
     \eda
     \item Then, $d_{t_3}^1$ is finite which concludes the proof for this case.
 \end{enumerate}
\qed
\end{proof}

The result no longer holds if we replace
the assumption $\concPart{\deriv{r}{w}} \semle \varepsilon$, for $w \in \Sigma^*$
by a simpler assumption like $\concPart{r} \semle \varepsilon$.
For example, consider the behavior $(x \conc \forkEff{y})^*$
where  $\concPart{x \conc \forkEff{y}} = \phi \semle \varepsilon$.
However, the set of dissimilar descendants of $(x \conc \forkEff{y})^*$ is infinite
as shown by the calculation
\bda{ll}
   & (x \conc \forkEff{y})^*
\\ \stackrel{x}{\rightarrow} & (\varepsilon \conc \forkEff{y}) \conc (x \conc \forkEff{y})^*
\\ \eqsim &             \forkEff{y} \conc (x \conc \forkEff{y})^*
\\ \stackrel{x}{\rightarrow} & \forkEff{\phi} \conc (x \conc \forkEff{y})^*
              + \forkEff{y} \conc \forkEff{y} \conc (x \conc \forkEff{y})^*
\\ \eqsim & \forkEff{y} \conc \forkEff{y} \conc (x \conc \forkEff{y})^*
\\ ...
\eda
The example also shows that the assumption
$\concPart{\deriv{r}{w}} \semle \varepsilon$, for $w \in \Sigma^*$
is necessary and cannot be weakened to words $w$ of a fixed length.

As an example, consider the behavior
$t = (x_1 \conc ... \conc x_n \conc x_{n+1} \forkEff{y})^*$
where for all $w \in \Sigma^*$ with length less or equal $n$
we find that $\concPart{\deriv{x_1 \conc ... \conc x_n \conc x_{n+1} \forkEff{y}}{w}} \semle \varepsilon$.
Via a similar calculation as above, we can show that
the set of dissimilar descendants of $t$ is infinite.

\section{Related Work}
\label{sec:related-work}

Shuffle expressions are regular expressions with operators for shuffle and shuffle closure.
Shaw \cite{DBLP:journals/tse/Shaw78} proposes to describe the behavior of software using flow
expressions, which extend shuffle expressions with further operators.
Gischer \cite{Gischer:1981:SLP:358746.358767} shows that shuffle expression generate
context-sensitive languages and proposes a connection to Petri net languages. 

The latter connection is made precise by
Garg and Ragunath \cite{Garg:1992:CRE:132251.132252}, who study concurrent regular expressions (CRE), which
are shuffle expressions extended with synchronous composition. They 
show that the class of CRE languages is equal to the class of Petri net languages.
The proof requires the presence of synchronous composition. Forkable expressions do not support
synchronous composition, but they are equivalent to \emph{unit expressions}, which are also defined
by Garg and Ragunath and shown to be strictly less powerful than CREs.



Warmuth and Haussler~\cite{DBLP:journals/jcss/WarmuthH84} present more refined complexity results
for the languages generated by shuffle
expressions. Jedrzejowicz~\cite{DBLP:journals/ipl/Jedrzejowicz87} shows that the nesting of iterated
closure operators matters.


\section*{Acknowledgments}

We thank the reviewers for their comments.

\bibliography{main}

\appendix

\newpage

\section{Supplementary Material}

\subsection{Well-definedness of $L$}
\label{sec:well-definedness-l}

For each $r$, the mapping $K \mapsto L (r, K)$ is a monotone mapping
in $\Power (\Sigma^*) \to \Power (\Sigma^*)$. The proof is by induction
on $r$.

\textbf{Case }$\phi$: $K \mapsto L (\phi, K) = \emptyset$ is monotone.

\textbf{Case }$\varepsilon$: $K\mapsto L (\varepsilon, K) = K$ is
monotone.

\textbf{Case }$x$: $K \mapsto L (x, K) = x \conc K$ is monotone.

\textbf{Case }$r+s$: $K \mapsto L (r+s, K) = L (r, K) \cup L (s, K)$
is monotone by induction.

\textbf{Case }$r\conc s$: $K \mapsto L (r\conc s, K) = L (r, L (s,
K))$ is monotone by induction.

\textbf{Case }$r^*$: $K \mapsto L (r^*, K) = \mu X. L (r, X) \cup
K$. By induction, $f = X \mapsto L (r, X)$ is monotone. Hence, $g_K = X
\mapsto L (r, X) \cup K$ is monotone, for all $K$. Furthermore, $f \le
g_K$ and $g_K \le g_L$ in the pointwise ordering of set-valued
functions, whenever $K \subseteq L$. It follows that $\mu g_K \le \mu
g_L$, which proves that $K \mapsto L (r^*,K)$ is monotone.

\textbf{Case }$\forkEff r$: $K \mapsto L (\forkEff r, K) = L (r) \|
K$. Immediate by induction and because $\|$ is monotone.

\subsection{Proof of Theorem~\ref{le:conc-seq-split}}

\begin{proof}
  By induction on $r$.

  \textbf{Case $\phi$, $\varepsilon$, $x$, $\forkEff r$:} immediate.

  \textbf{Case $r+s$:} immediate by induction.

  \textbf{Case $r \conc s$:}~\begin{minipage}[t]{0.7\linewidth}
    \vspace{-1.8\baselineskip}
    \begin{align*}
      \concPart{r \conc s} + \seqPart{r \conc s}
      &= \concPart{r} \conc \concPart{s} + \concPart{ r}
        \conc \seqPart{s} + \seqPart{r} \conc s \\
      &\semeq \concPart{r} \conc (\concPart{s} + \seqPart{s}) +
        \seqPart{r} \conc s \\
      &\stackrel{IH}\semeq \concPart{r} \conc s + \seqPart{r} \conc s \\
      &\semeq (\concPart{r} + \seqPart{r}) \conc s 
      \\&
          \stackrel{IH}\semeq r \conc s \\
    \end{align*}
  \end{minipage}
  
  \textbf{Case $r^*$:}~\begin{minipage}[t]{0.7\linewidth}
    \vspace{-1.8\baselineskip}
    \begin{align*}
      \concPart{r^*} + \seqPart{r^ *}
      &= \concPart{r}^* + \concPart{r}^* \conc
        \seqPart{r} \conc r^*\\
      &\stackrel{IH}\semeq \concPart{r}^* + \concPart{r}^* \conc
        \seqPart{r} \conc (\concPart{r} + \seqPart{r})^*\\
      &\semeq (\concPart{r} + \seqPart{r})^*
      \\ &
           \stackrel{IH}\semeq r^*
    \end{align*}
  \end{minipage}

  \qed
\end{proof}

\subsection{Proof of Theorem~\ref{theorem:fork-free-regular}}

\begin{proof}
  Induction on $r$ where we generalize the statement
  to $L(r, K) = \Sem{r} \conc K$, for all $K\subseteq\Sigma^*$,
  and use Arden's lemma for the case $r^*$.
  The theorem follows by setting $K=\{\varepsilon\}$.
  
  \textbf{Cases }$\phi$, $\varepsilon$, $x$: immediate.

  \textbf{Case }$r+s$: immediate by IH.

  \textbf{Case }$r\conc s$: 
  $
  L (r \conc s, K) = L (r, L (s, K)) = \Sem{r} \conc (\Sem{s} \conc K) = \Sem{r \conc s} \conc K
  $.

  \textbf{Case }$r^*$:
  \bda{ll}
  & L (r^*, K)
  \\ = & \mbox{(by definition)}
  \\   & \mu X. L (r, X) \cup K
  \\ = & \mbox{(by IH, $X$ must be in $\Sigma^*$)}
  \\   & = \mu X. \Sem{r} \conc X \cup K
  \\   & \mbox{(by Arden's lemma)}
  \\   & = \Sem{r^*} \conc K
  \eda
  \qed
\end{proof}

\subsection{Proof of Theorem~\ref{th:behaviors-kleene-algebra}}

\begin{proof}
Following~\cite{Kozen:1990:KAC:645720.663687}
a Kleene algebra satisfies the axioms of an idempotent semiring
and the following axioms:
\begin{enumerate}
  \item $\varepsilon + r \conc r^* \semle r^*$
  \item $\varepsilon + r^* \conc r \semle r^*$
  \item If $r \conc s \semle s$ then $r^* \conc s \semle s$
  \item If $s \conc r \semle s$ then $s \conc r^* \semle s$
\end{enumerate}

It is straightforward to verify that behaviors form on idempotent
semiring.

Axioms (1) and (2) follow by calculation from the semantics.

Consider (3). Suppose $r \conc s \semle s$.
We will show that by induction for any $n$ we have that
$r^n \conc s \semle s$ where
$r^0 = \varepsilon$ and $r^{n+1} = r \conc r^n$.
Case $n=0$ holds immediately.
For case $n+1$ we perform the following calculations.
\bda{ll}
     & r^{n+1} \conc s 
\\ = & r \conc r^n \conc s  
\\   & \mbox{(IH)}
\\ \semle & r \conc s
\\ \semle & s
\eda
Thus, we find that $r^* \conc s \semle s$.
Case (4) can be verified similarly.
\qed
\end{proof}

\subsection{Proof of  Lemma~\ref{le:basic-seq-conc-props}}
\label{sec:proof-lemma-refl}

\begin{proof}
  \begin{enumerate}
  \item Trivial induction.
  \item Induction on $r$. The cases for $\phi$, $\varepsilon$, $l$, $r+s$, $\forkEff r$ are trivial
    or immediate by the inductive hypothesis.
    \begin{align*}
      \concPart{\seqPart{r \conc s}}
      & = \concPart{\seqPart{r} \conc s +\concPart{ r} \conc \seqPart{s}} \\
      & = \concPart{\seqPart{r} \conc s} +\concPart{\concPart{ r} \conc \seqPart{s}} \\
      & = \concPart{\seqPart{r}} \conc \concPart{s} +\concPart{\concPart{ r}} \conc \concPart{\seqPart{s}} \\
      & \semeq \phi \conc \concPart{s} +{\concPart{ r}} \conc \phi \\
      & \semeq \phi
    \end{align*}
    \begin{align*}
      \concPart{\seqPart{r^*}}
      & = \concPart{\concPart r^* \conc \seqPart r \conc r^*} \\
      & = \concPart{\concPart r^*} \conc \concPart{\seqPart r} \conc \concPart{r^*}   \\
      & \semeq \concPart{\concPart r}^* \conc \phi \conc \concPart{r}^* \\
      & \semeq \phi
    \end{align*}
  \item Induction on $r$. The cases for $\phi$, $\varepsilon$, $l$, $r+s$, $\forkEff r$ are trivial or immediate
    by inductive hypothesis.
    \begin{align*}
      \seqPart{\concPart{r \conc s}}
      & = \seqPart{\concPart{r} \conc \concPart{s}} \\
      & = \seqPart{\concPart{r}} \conc \concPart{s} + \concPart{\concPart{r}} \conc
        \seqPart{\concPart{s}} \\
      & \semeq \phi \conc \concPart{s} + {\concPart{r}} \conc \phi \\
      & \semeq \phi
    \end{align*}
    \begin{align*}
      \seqPart{\concPart{r^*}}
      & = \seqPart{\concPart r^*} \\
      & = \concPart{\concPart r}^* \conc \seqPart{\concPart r} \conc \concPart r^*  \\
      & \semeq {\concPart r}^* \conc \phi \conc \concPart r^* \\
      & \semeq \phi
    \end{align*}
  \item Induction on $r$. The cases for $\phi$, $\varepsilon$, $l$, $r+s$, $\forkEff r$ are trivial or immediate
    by inductive hypothesis.
    \begin{align*}
      \seqPart{\seqPart{r \conc s}}
      & = \seqPart{\seqPart{r} \conc s +\concPart{ r} \conc \seqPart{s}} \\
      & =  \seqPart{\seqPart{r} \conc s} +\seqPart{\concPart{ r} \conc \seqPart{s}} \\
      & = \seqPart{\seqPart{r}} \conc s + \concPart{\seqPart{r}} \conc \seqPart{s}
        + \seqPart{\concPart{r}} \conc \seqPart{s} + \concPart{\concPart{r}} \conc
        \seqPart{\seqPart{s}} \\
      & \{\text{by 1., 2., 3., and the inductive hypothesis}\}\\
      & \semeq {\seqPart{r}} \conc s  + {\concPart{r}} \conc {\seqPart{s}} 
    \end{align*}  
    \begin{align*}
      \seqPart{\seqPart{r^*}}
      & = \seqPart{\concPart r^* \conc \seqPart r \conc r^*} \\
      & = \seqPart{\concPart r^*} \conc \seqPart r \conc r^* + \concPart{\concPart r^*} \conc \seqPart{ \seqPart r \conc r^*} \\
      & = \seqPart{\concPart r^*} \conc \seqPart r \conc r^*
        + {\concPart r}^* \conc (\seqPart{\seqPart r} \conc r^* + \concPart{ \seqPart r} \conc
        \seqPart{r^*})  \\
      & = \seqPart{\concPart r^*} \conc \seqPart r \conc r^*
        + {\concPart r}^* \conc ({\seqPart r} \conc r^* + \phi \conc
        \seqPart{r^*})  \\
      & = \seqPart{\concPart r^*} \conc \seqPart r \conc r^*
        + {\concPart r}^* \conc ({\seqPart r} \conc r^*)  \\
      & = \concPart{\concPart r}^* \conc \seqPart{\concPart r} \conc \concPart r^* \conc \seqPart r \conc r^*
        + {\concPart r}^* \conc ({\seqPart r} \conc r^*)  \\
      & \semeq {\concPart r}^* \conc \phi \conc \concPart r^* \conc \seqPart r \conc r^*
        + {\concPart r}^* \conc ({\seqPart r} \conc r^*) \\
      & \semeq {\concPart r}^* \conc ({\seqPart r} \conc r^*)
    \end{align*}
  \end{enumerate}
  \mbox{}
  \qed
\end{proof}

\subsection{Proof of Lemma \ref{le:concpart-nullable}}

\begin{proof}
  If $\varepsilon \semle \concPart{r}$ and by decomposition
  $r \semeq \seqPart{r} + \concPart{r}$ we find that $\varepsilon \in L(r)$.

  The other direction requires induction on $r$.  For brevity,
we only consider some cases.
\bda{ll} & \varepsilon \in L(r \conc s) \\ \implies & \varepsilon \in L(r) \wedge \varepsilon \in
L(s) \\ \stackrel{\mbox{IH}}{\implies} & \varepsilon \semle \concPart{r} \wedge \varepsilon \semle
\concPart{s} \\ \implies & \varepsilon \semle \concPart{r} \conc \concPart{s} = \concPart{r \conc s}
\eda

For Kleene star, by definition $\concPart{r^*} = \concPart{r}^*$ and  $\varepsilon \semle
\concPart{r}^*$ is a derived property in a Kleene algebra. Thus, $\varepsilon \semle \concPart{r^*}$ follows immediately.
\qed
\end{proof}

\subsection{Auxiliary Statements}

\begin{lemma}\label{lemma:concurrent-concatenation}
Let $r, s$ be behaviors.
Then
  $L (\concPart{r}) \| L(\concPart{s}) = L (\concPart{r \conc s})$.
\end{lemma}
\begin{proof}
  \begin{align*}
    L (\concPart{r}) \| L(\concPart{s})
    & = L (\concPart{r}, L(\concPart{s})) \\
    & = L (\concPart{r}, L(\concPart{s}, \{\varepsilon\})) \\
    & = L (\concPart{r \conc s}, \{\varepsilon\}) \\
    & = L (\concPart{r \conc s}) 
  \end{align*}
\qed
\end{proof}

\begin{lemma}
\label{le:split-concpart-async}
Let $r$ be a behavior and $K$ a language.
Then
  $L (\concPart r, K) = L (\concPart r) \| K$
\end{lemma}
\begin{proof}
  Induction on $f$.

  \textbf{Case }$\phi$.

  Trivial as $\concPart \phi = \phi$, $L (\phi, K) = \emptyset = \emptyset \| K$.

  \textbf{Case }$\varepsilon$.

  $\concPart \varepsilon  = \varepsilon$. $L (\varepsilon, K) = K = \{\varepsilon\} \| K$.

  \textbf{Case }$x$.

  $\concPart x = \phi$. $L (\phi, K) = \emptyset = \emptyset \| K$.

  \textbf{Case }$r+s$.
  \begin{align*}
    L (\concPart{r+s}, K) 
    & = L (\concPart{r}, K) \cup L (\concPart{s}, K) \\
    & = L (\concPart{r}) \| K \cup L (\concPart{s}) \| K \\
    & = L (\concPart{r+s}) \| K \\
  \end{align*}

  \textbf{Case }$r \conc s$. (requires Lemma~\ref{lemma:concurrent-concatenation})
  \begin{align*}
    L (\concPart{r \conc s}, K) 
    & = L (\concPart{r}\conc \concPart{s}, K) \\
    & = L (\concPart{r}, L(\concPart{s}, K)) \\
    & = L (\concPart{r}) \| L(\concPart{s}) \| K \\
    & = L (\concPart{r \conc s}) \| K \\
  \end{align*}

  \textbf{Case }$r^*$.
  \begin{align*}
    L (\concPart{r^*}, K)
    & = L (\concPart r^*, K) \\
    & = \mu X. (L (\concPart r, X) \cup K) \\
    & \{\text{IH}\} \\
    & = \mu X. (L (\concPart r) \| X \cup K) \\
    & =L (\concPart r)^\sharp \| K \\
    & = (\mu X. L (\concPart{r}) \| X \cup \{\varepsilon\}) \| K \\
    & \{\text{IH}\} \\
    & = (\mu X. L (\concPart{r}, X) \cup \{\varepsilon\}) \| K \\
    & = L (\concPart{r}^*, \{\varepsilon\}) \| K \\
    & = L (\concPart{r^*}) \| K \\
  \end{align*}

  \textbf{Case }$\forkEff r$.
  \begin{align*}
    L (\concPart{\forkEff r}, K)
    & = L ({\forkEff r}, K) \\
    & = L (r) \| K \\
    & = L (r) \| \{\varepsilon\} \| K \\
    & = L (\forkEff r, \{\varepsilon\}) \| K \\
    & = L (\concPart{\forkEff r}) \| K \\
  \end{align*}
\qed
\end{proof}

\subsection{Remaining cases of Theorem~\ref{th:deriv-correctness}}

\begin{proof}
  \textbf{Case }$\forkEff r$.
  \begin{align*}
    & L (\deriv{\forkEff r}{x}, K ) \cup L (\concPart{\forkEff r}) \| \leftQ x K \\
    & = L (\forkEff{\deriv{r}{x}}, K ) \cup L ({\forkEff r}) \| \leftQ x K \\
    & = L ({\deriv{r}{x}}) \| K  \cup L (r) \| \leftQ x K \\
    & \{\text{IH}\} \\
    & = \leftQ x {L (r)} \|  K \cup L (r) \| \leftQ x K\\
    & = \leftQ x {(L (r) \|  K)}\\
    & = \leftQ x {(L (\forkEff r, K))}
  \end{align*}

  \textbf{Case }$\phi$. $\concPart\phi = \phi$.

  $L (\deriv{\phi}{x}, K) = L (\phi, K) = \emptyset = \leftQ x {L (\phi , K)} $

  \textbf{Case }$\varepsilon$.  $\concPart \varepsilon = \varepsilon$.

  $L (\deriv{\varepsilon}{x}, K) \cup L (\concPart \varepsilon) \| (\leftQ x K) = L (\phi, K)  \cup \leftQ x K =  \leftQ x K = \leftQ x {L (\varepsilon, K)}$

  \textbf{Case }$x$. $\concPart x = \phi$.

  $L (\deriv{x}{x}, K) = L (\varepsilon, K) = K = \leftQ x {L (x, K)}$

  \textbf{Case }$y \ne x$.

  $L (\deriv{y}{x}, K) = L (\phi, K) = \emptyset = \leftQ x {L (y, K)}$

  \textbf{Case }$r+s$. $\concPart{r+s} = \concPart{r} + \concPart{s}$.
  \begin{align*}
    &L (\deriv{r+x}{x}, K) \cup L(\concPart{r+s}) \| (\leftQ x K) \\
    & =L (\deriv{r}{x}, K)  \cup L(\concPart{r}) \| (\leftQ x K) \cup L (\deriv{s}{x}, K) \cup
      L(\concPart{s}) \| (\leftQ x K) \\
    & \{\text{IH}\} \\
    & = \leftQ x {L (r, K)} \cup \leftQ x {L (s, K)}\\
    & = \leftQ x {L (r + s, K)}
  \end{align*}

  \textbf{Case }$r \conc s$. $\concPart{r \conc s} = \concPart{r} \conc \concPart{s}$.
  \begin{align*}
    & L (\deriv{r \conc s}{x}, K) \cup L(\concPart{r \conc s}) \| (\leftQ x K) \\
    & = L (\deriv{r}{x} \conc s, K) \cup L(\concPart{r} \conc \deriv{s}{x}, K) \cup L
      (\concPart{s}) \| L (\concPart{s}) \| (\leftQ x K) \\
    & = L (\deriv{r}{x}, L(s, K)) \cup L(\concPart{r}, L( \deriv{s}{x}, K)) \cup
      L(\concPart{r}) \| L (\concPart{s}) \| (\leftQ x K) \\
    & = L (\deriv{r}{x}, L(s, K)) \cup L(\concPart{r}) \| L( \deriv{s}{x}, K) \cup
      L(\concPart{r}) \| L (\concPart{s}) \| (\leftQ x K) \\
    & = L (\deriv{r}{x}, L(s, K)) \cup L(\concPart{r}) \| \big( L( \deriv{s}{x}, K) \cup
      L (\concPart{s}) \| (\leftQ x K)\big) \\
    & \{\text{IH}\} \\
    & = L (\deriv{r}{x}, L( s, K)) \cup L (\concPart{r})\| (\leftQ x {L( s, K)}) \\
    & \{\text{IH}\} \\
    & = \leftQ x {L (r, L( s, K))} \\
    & = \leftQ x {L (r \conc s, K)}
  \end{align*}  

 \textbf{Case }$r^*$.

We verify for all $n \leq 0$ that
\bda{c}
 \leftQ x L(r^n,K) = L(\deriv{r^n}{x},K) \cup L(\concPart{r^n}) \| (\leftQ x K)
\eda
where $r^0 = \varepsilon$ and $r^{n+1} = r \conc r^n$.

\textbf{Case $n=0$:} Straightforward

\textbf{Case $n \implies n+1$:}
\bda{ll}
     & L(\deriv{r^{n+1}}{x},K) \cup L(\concPart{r^{n+1}}) \| (\leftQ x K)
\\ = & L(\deriv{r}{x} \conc r^n + \concPart{r} \conc \deriv{r^n}{x},K) 
       \cup L(\concPart{r^{n+1}}) \| (\leftQ x K)
\\ = & L(\deriv{r}{x}, L(r^n,K)) \cup
       L(\concPart{r},L(\deriv{r^n}{x},K)) \cup 
       L(\concPart{r^{n+1}}) \| (\leftQ x K)
\\ = & L(\deriv{r}{x}, L(r^n,K)) \cup
       L(\concPart{r},L(\deriv{r^n}{x},K)) \cup 
       L(\concPart{r}, L(\concPart{r^n})) \| (\leftQ x K)
\\ & \mbox{(Lemma~\ref{le:split-concpart-async})}
\\ = & L(\deriv{r}{x}, L(r^n,K)) \cup
       L(\concPart{r}) \| L(\deriv{r^n}{x},K) \cup
       L(\concPart{r} \| L(\concPart{r^n}) \| (\leftQ x K)
\\ = & L(\deriv{r}{x}, L(r^n,K)) \cup
       L(\concPart{r}) \| (L(\deriv{r^n}{x},K) \cup
                          L(\concPart{r^n}) \| (\leftQ x K)) 
\\   & \mbox{(IH)}
\\ = & L(\deriv{r}{x}, L(r^n,K)) \cup
       L(\concPart{r}) \| (\leftQ x (L(r^n,K)))
\\   & \mbox{(IH)}
\\ = & \leftQ x L(r, L(r^n,K))
\\ = & \leftQ x L(r^{n+1},K)
\eda

Then, the actual statement
\bda{c}
 \leftQ l L(r^*,K) = L(\deriv{r^*}{x},K) \cup L(\concPart{r^*}) \| (\leftQ x K)
\eda
follows from the fact
that
\bda{l}
  (\leftQ x L(r^n,K)) \leq (\leftQ x L(b^*,K))
\\
\\ L(\deriv{r^n}{x},K) \cup L(\concPart{r^n}) \| (\leftQ x K) \leq 
   L(\deriv{r^*}{x},K) \cup L(\concPart{r^*}) \| (\leftQ x K)
\eda
and by showing that for all finite words $w \in \Sigma^*$
we have that $w \in (\leftQ x L(r^*,K))$ iff 
$w \in L(\deriv{r^*}{x},K) \cup L(\concPart{r^*}) \| (\leftQ x K)$.

If $w \in (\leftQ x L(r^*,K))$
then $w \in (\leftQ x L(b^n,K))$ for some $n\leq 0$.
From above, we conclude that
$w \in L(\deriv{r^n}{x},K) \cup L(\concPart{r^n}) \| (\leftQ x K) \leq 
   L(\deriv{r^*}{x},K) \cup L(\concPart{r^*}) \| (\leftQ x K)$.
The argument is similar for the other direction.
\qed
\end{proof}

\subsection{Proof of Lemma~\ref{le:eqsim-eps-phi}}

\begin{proof}
By induction on $r$.
Consider $L(r) = \{ \varepsilon \}$.

Cases $x$ and $\phi$ do not apply.
Case $\varepsilon$ is straightforward.

\textbf{Case $r+s$:}
$L(r + s) = \{ \varepsilon \}$ implies
that $L(r) \subseteq \{ \varepsilon \}$ and $L(s) \subseteq \{ \varepsilon \}$.
There are four subcases to consider.
Suppose $L(r) = \{ \varepsilon \}$ and $L(s) = \{ \}$.
By IH, $r \eqsim \varepsilon$ and either $s \eqsim \phi$.
Then, $r + s \eqsim \varepsilon$, by (Unit). Other cases are similar.

\textbf{Case $r \conc s$:} $L(r \conc s) = L(r) \conc L (s) = \{ \varepsilon \}$ implies
that $L(r) = \{ \varepsilon \}$ and $L(s) = \{ \varepsilon \}$.
By IH, $r \eqsim \varepsilon$ and $s \eqsim \varepsilon$.
Hence, $r \conc s \eqsim \varepsilon$, by (Empty Word).

\textbf{Case $r^*$:}
$L(r^*) = \{ \varepsilon \}$ implies that $L(r) = \{ \varepsilon \}$.
By IH, $r \eqsim \varepsilon$ from which we can conclude that
$r^* \eqsim \varepsilon$.

\textbf{Case $\forkEff{r}$:} Similar to the above.

Consider $L(r) = \{ \}$.
Cases $x$, $\varepsilon$ and $r^*$ do not apply. Case $\phi$ is straightforward.

\textbf{Case $r \conc s$:}
$L(r \conc s) = \{ \}$ implies that $L(r) = \{ \} \vee L(s) = \{ \}$.
Suppose $L(r) = \{ \}$. By IH, $r \eqsim \phi$ and therefore
$r \conc s \eqsim \phi$. The other case is similar.

\textbf{Case $r + s$:} Similar to the above.

\textbf{Case $\forkEff{r}$:}
$L(\forkEff{r}) = \{ \}$ implies $L(r) = \{ \}$.
Via similar arguments as above we find that $\forkEff{r} \eqsim \phi$.
\qed
\end{proof}

\subsection{Proof of Lemma \ref{le:c-d-c=d-c}}

\begin{proof}
  \textbf{Case }$r = \phi, \varepsilon, x$: trivial as  $\concPart{r} \syneq \phi$.

  \textbf{Case }$r+ s$: immediate from the IH.

  \textbf{Case }$r \conc s$: 
  \begin{align*}
    & \deriv{\concPart{r \conc r}}{x} \\
    &\syneq \deriv{\concPart{r}\conc    \concPart{x}}{x} \\
    &\syneq \deriv{\concPart{r}}{x} \conc \concPart{s} + \concPart{\concPart{r}} \conc
       \deriv{\concPart{x}}{x} \\
    & \qquad\{\text{IH, Lemma~\ref{le:basic-seq-conc-props} part 1}\} \\
    &\syneq \concPart{\deriv{\concPart{r}}{x}} \conc \concPart{\concPart{s}} +
      \concPart{\concPart{\concPart{r}}} \conc \concPart{\deriv{\concPart{s}}{x}} \\
    &\syneq \concPart{\deriv{\concPart{r \conc s}}{x}}
  \end{align*}

  \textbf{Case }$r^*$:
  \begin{align*}
    & \deriv{\concPart{r^*}}{l} \\
    & \syneq \deriv{\concPart{r}^*}{x} \\
    & \syneq \deriv{\concPart{r}}{x} \conc \concPart{r}^* \\
    & \qquad\{\text{IH, Lemma~\ref{le:basic-seq-conc-props} part 1}\} \\
    & \syneq \concPart{\deriv{\concPart{r}}{x}} \conc \concPart{\concPart{r}}^* \\
    & \syneq \concPart{\deriv{\concPart{r^*}}{x}}
  \end{align*}

  \textbf{Case }$\forkEff r$:
  \begin{align*}
    & \deriv{\concPart{\forkEff r}}{x} \\
    & \syneq \deriv{{\forkEff r}}{x} \\
    & \syneq \forkEff{\deriv{{ r}}{x}} \\
    & \syneq \concPart{\forkEff{\deriv{{ r}}{x}}} \\
    & \syneq \concPart{\deriv{\concPart{\forkEff r}}{x}} 
  \end{align*}
\qed
\end{proof}


\subsection{Proof of Theorem \ref{th:finiteness-well-behaved}}

\begin{proof}
 We need to generalize the statement to obtain the result:
 If $t$ is well-behaved then
 $\Card d_t^i < \infty$ for all $i \geq 0$
   where $d_t^0 = \deriv{t}{\eqsim}$ and $d_t^{n+1} = \deriv{\concPart{d_t^n}}{\eqsim}$.

 Based on Lemmas~\ref{le:deriv-finite-fp} and \ref{le:drop-conc-deriv-conc-deriv} 
we find that $d_t^{n+1} = d_t^n$ for $n \geq 1$.
That is, in the induction step it is sufficient to establish
that $d_t^0$ and $d_t^1$ are finite.
We proceed by induction on $t$.

\noindent \textbf{Cases $\varepsilon$, $x$, $\phi$:} Straightforward.

\noindent \textbf{Case $r+s$:} By the IH, $d_r^i$, $d_s^i$ are finite for any $i \geq 0$.
We know that $\deriv{r+s}{} = \deriv{r}{} + \deriv{s}{}$.
The above can be represented by
 $\deriv{r}{\eqsim} + \deriv{s}{\eqsim}$. Immediately, we find
 $d_{r+s}^0$ is finite.

Via similar reasoning, we can show
that the $\deriv{\concPart{\deriv{r+s}{}}}{}$
can be represented by 
$\deriv{\concPart{\deriv{r}{\eqsim}}}{\eqsim} + \deriv{\concPart{\deriv{s}{\eqsim}}}{\eqsim}$. Hence, $d_{r+s}^1$ which concludes the proof for this case.

\noindent \textbf{Case $\forkEff{r}:$} By the IH, $d_r^i$ are finite for any $i \geq 0$.

From $\deriv{\forkEff{r}}{} = \forkEff{\deriv{r}{}}$
and the above we can conclude that $d_{\forkEff{r}}^0$ is finite.

To establish $d_{\forkEff{r}}^1$ is finite, we apply the following reasoning.
\bda{ll}
  & \deriv{\concPart{\deriv{\forkEff{r}}{\eqsim}}}{\eqsim}
\\ \eqsim & \deriv{\concPart{\forkEff{\deriv{r}{\eqsim}}}}{\eqsim}
\\ \eqsim & \deriv{\forkEff{\deriv{r}{\eqsim}}}{\eqsim}
\\ \eqsim & \forkEff{\deriv{\deriv{r}{\eqsim}}{\eqsim}}
\\ \eqsim & \forkEff{\deriv{r}{\eqsim}}
\eda
The set $\deriv{r}{\eqsim}$ is finite.
Hence, $d_{\forkEff{r}}^1$ is finite.

\noindent \textbf{Case $r^*$:} By the IH, $d_r^i$ are finite for any $i \geq 0$.
We show that $d_{r^*}^0$ is finite.
 \begin{enumerate}
   \item By Lemma~\ref{le:derivatives-of-r*}, descendants of $r^*$ are of the form
   \bda{c}
      s_1 \conc ... \conc s_n \conc \deriv{r}{} \conc r^*
   \eda
   where $s_i \in \deriv{\concPart{\deriv{r}{w}}}{}$ for some $w\in \Sigma^+$.
 \item From Lemma~\ref{le:eqsim-eps-phi} and 
     the assumption that all subterms of the form $r^*$ have the
  property that $\concPart{\deriv{r}{w}} \le \varepsilon$, for $w \in \Sigma^*$,
    we conclude that the above is either similar to $\phi$ or $d_r^0 \conc r^*$.
  \item Hence, $d_{r^*}^0$ is finite.
 \end{enumerate}
  
 Next, we show that $d_{r^*}^1$ is finite.
  \begin{enumerate}
   \item We observe the possible forms of (dissimilar) descendants of $r^*$.
   \item For case $\phi$ we immediately find that $\deriv{\concPart{\phi}}{\eqsim}$ is finite.
   \item For case $d_r^0 \conc r^*$, we consider the possible descendants 
         of $\concPart{d_r^0 \conc r^*} = \concPart{d_r^0} \conc \concPart{r^*}$.
   \item By Lemma ~\ref{le:form-derivative-concatenation}, the shape of such terms
        is of the form
        \bda{c}
            \deriv{\concPart{d_r^0}}{} \conc \concPart{r^*} 
          + \concPart{d_r^0} \conc \deriv{\concPart{r^*}}{}
          + \sum \deriv{\concPart{\deriv{\concPart{d_r^0}}{}}}{} \conc \deriv{\concPart{r^*}}{}
        \eda
     \item By Lemma~\ref{le:deriv-rep-eqsim} the above is similar to
             \bda{c}
             d_r^1 \conc \concPart{r^*} 
          + \concPart{d_r^0} \conc \deriv{\concPart{r^*}}{}
          + \sum d_r^2 \conc \deriv{\concPart{r^*}}{}
        \eda
     \item We know that $d_r^i$ are finite. What remains is to show that
           $\deriv{\concPart{r^*}}{\eqsim}$ is finite.
     \item Via similar reasoning as above we can argue that the descendants of
           $\concPart{r^*} = \concPart{r}^*$ are of the form
       \bda{c}
         t_1 \conc ... \conc t_m \conc \deriv{\concPart{r}}{} \conc \concPart{r}^*
       \eda
        where $t_i \in \deriv{\concPart{\deriv{\concPart{r}}{w}}}{}$ for some $w \in \Sigma^+$.
 
     \item Each $t_i$ is either $\phi$ or $\varepsilon$ based on the following approximation.
          First, we find that 
         $\concPart{\deriv{\concPart{r}}{w}} \leq \concPart{\deriv{\concPart{r} + \seqPart{r}}{w}} = \concPart{\deriv{r}{w}}$. 
         As we know that $\concPart{\deriv{r}{w}} \leq \varepsilon$ for $w \in \Sigma^+$,
         we can conclude that $t_i \leq \varepsilon$.
     \item Hence, descendants of $\concPart{r}^*$ are either similar to $\phi$
           or terms of the form $\deriv{\concPart{r}}{} \conc \concPart{r}^*$.
     \item As we know $d_r^1$ is finite and subsumes $\deriv{\concPart{r}}{}$
           we can conclude that $\deriv{\concPart{r^*}}{\eqsim}$ is finite.
           Thus, we are done.
  \end{enumerate}

\noindent \textbf{Case $r \conc s$:} By the IH, $d_r^i$ and $d_s^i$ are finite for any $i \geq 0$.
 We show that $d_{r \conc s}^0$ is finite.
 \begin{enumerate}
    \item By Lemma~\ref{le:form-derivative-concatenation}, each derivative $\deriv{r \conc s}{w}$ has the form
       $$
        \deriv{r}{} \conc s + \concPart{r} \conc \deriv{s}{} + \sum \deriv{\concPart{\deriv{r}{}}}{} \conc \deriv{s}{}
       $$
    \item By Lemma~\ref{le:deriv-rep-eqsim}
          the above is similar to         
       $$
        \deriv{r}{\eqsim} \conc s + \concPart{r} \conc \deriv{s}{\eqsim} + \sum \deriv{\concPart{\deriv{r}{\eqsim}}}{\eqsim} \conc \deriv{s}{\eqsim}
       $$ 
    \item Immediately, we can conclude that $d_{r \conc s}^0$ is finite.
 \end{enumerate}
  Next, we consider $d_{r \conc s}^1$.
 \begin{enumerate}
   \item From above, the form of (dissimilar) descendants of $r \conc s$ is
       $$
        \underbrace{\deriv{r}{} \conc s}_{t_1} + 
        \underbrace{\concPart{r} \conc \deriv{s}{}}_{t_2} + 
       \sum \underbrace{\deriv{\concPart{\deriv{r}{}}}{} \conc \deriv{s}{\eqsim}}_{t_3}
       $$ 
      For each $t_i$ we will show that $d_{t_i}^1$ is finite and thus follows
      the desired result.
   \item By Lemma~\ref{le:form-derivative-concatenation}, descendants of $\concPart{t_1}$ are of the form
    $$
   \deriv{\concPart{\deriv{r}{}}}{} \conc \concPart{s}
   + \concPart{\deriv{r}{}} \conc \deriv{\concPart{s}}{}
   + \sum \deriv{\concPart{\deriv{\concPart{\deriv{r}{}}}{}}}{} \conc \deriv{\concPart{s}}{}
    $$
 \item As we know by the IH, $d_r^i$ and $d_s^i$ are finite for any $i \geq 0$.
       Hence, via similar reasoning as above we can conclude that the above
      is similar to expressions of the form
     \bda{c}
         d_r^1 \conc \concPart{s} + \concPart{d_r^0} \conc d_s^1
       + \sum d_r^2 \conc d_s^1
     \eda
      As all sub-components are finite,
       we can immediately conclude that $d_{t_1}^1$ is finite.
  \item We consider the descendants of $\concPart{t_2}$ which are of the following form
  \bda{c}
   \deriv{\concPart{r}}{} \conc \concPart{\deriv{s}{}}
   + \concPart{r} \conc \deriv{\concPart{\deriv{s}{}}}{}
   + \sum \deriv{\concPart{\deriv{\concPart{r}}{}}}{} \conc \deriv{\concPart{\deriv{s}{}}}{}
 \eda
   \item The above is similar to
   \bda{c}
       d_r^1 \conc \concPart{d_s^0}
     + \concPart{r} \conc d_r^1
     + \sum d_r^2 \conc d_s^1
   \eda
    Thus, we find that $d_{t_2}^1$ is finite.
   \item Finally, we observe that shape of descendants of $\concPart{t_3}$
    \bda{l}
        \deriv{\concPart{\deriv{\concPart{\deriv{r}{}}}{}}}{} \conc \concPart{\deriv{s}{}}
  + \concPart{\deriv{\concPart{\deriv{r}{}}}{}} \conc \deriv{\concPart{\deriv{s}{}}}{}
 \\  + \sum \deriv{\concPart{\deriv{\concPart{\deriv{\concPart{\deriv{r}{}}}{}}}{}}}{}
         \conc \deriv{\concPart{\deriv{s}{}}}{}
    \eda
    \item The above is similar to
     \bda{c}
          d_r^2 \conc \concPart{d_s^0}
       +  \concPart{d_r^1} \conc d_s^1
       + \sum d_r^3 \conc d_s^1
     \eda
     \item Then, $d_{t_3}^1$ is finite which concludes the proof for this case.
 \end{enumerate}
\qed
\end{proof}

\subsection{Similarity is decidable}

\begin{figure}[tp]
\bda{c}

 \rlabel{Comm} ~ \myirule{s < r}{r + s \rightarrow s + r}
\\ \\
 \rlabel{Assoc1} ~ \myirule{s < r}{r + (s + t) \rightarrow s + (r + t)}
\\ \\
 \rlabel{Assoc2} ~ (r + s) + t \rightarrow r + (s + t)
\\ \\
 \rlabel{Idemp} ~ r + r \rightarrow r
\\ \\
  \rlabel{U} ~ \phi + r \rightarrow r
\ba{cc}
 \ba{cc}
    \rlabel{EW} &
    \ba{l}
            \varepsilon \conc r \rightarrow r
      \\ \\ r \conc \varepsilon \rightarrow r
      \\ \\ \varepsilon^* \rightarrow \varepsilon
      \\ \\ \forkEff{\varepsilon} \rightarrow \varepsilon
    \ea
 \ea
  &
 \ba{cc}
    \rlabel{EL} &
    \ba{l}
           \phi \conc r \rightarrow \phi
     \\ \\ r \conc \phi \rightarrow \phi
     \\ \\ \phi^* \rightarrow \varepsilon
     \\ \\ \forkEff{\phi} \rightarrow \phi
    \ea
 \ea
\ea
\eda
  \caption{Normalization}
  \label{fig:term-normalization}
\end{figure}

We turn similarity rules from Figure~\ref{fig:similarity}
into term rewriting rules.
Roughly, the right-hand side of $\eqsim$ is replaced by
the left-hand side. See Figure \ref{fig:term-normalization}.
Rules \rlabel{EW} and \rlabel{EL} are straightforward.
Alternatives are normalized into right-associative normal form.
See rule \rlabel{Assoc2}. To easily enforce
idempotence (rule \rlabel{Idemp}) we sort
terms in a sum according to their term size.
See rules \rlabel{Comm} and \rlabel{Assoc1}.
We assume that $|r|$ denotes the size of $r$.
We write $r < s$ iff $|r| < |s|$.
Finally, rule \rlabel{U} removes $\phi$ in a sum.
We assume that the term size of behaviors is
such that $|\phi| \leq |r|$ for any behavior $r$.
Thus, the single rule \rlabel{U} is sufficient.

\begin{lemma}
The rewrite system in Figure \ref{fig:term-normalization} 
is terminating and confluent.
\end{lemma}
\begin{proof}
Termination is due to the fact
that the he size of the left-hand side term becomes
smaller than the size of the right-hand side term.
The exceptions are rules \rlabel{Comm} and \rlabel{Assoc1-2}.
However, it is clear that there can only be a finite
number of applications of these terms. In case 
of \rlabel{Comm} and \rlabel{Assoc1} we shift smaller
terms to the 'left'. In case of \rlabel{Assoc2},
the left part of an alternative becomes smaller.

We show confluence by establishing local confluence.
That is, all critical pairs are confluence.
Given that the term rewrite system is terminating,
we obtain confluence by application of Newman's Lemma.

It is a straightforward exercise to verify that
all critical pairs are joinable.
For example, note the presence of rule \rlabel{Assoc2}.
Thus, the critical pair among \rlabel{Comm} and \rlabel{Assoc2}
becomes joinable.
\qed
\end{proof}

Based on the above result we can build canonical normal forms
by exhaustive rule application until we reach a fixpoint.
We write $r \rightarrow^* t$ to denote that $t$ is term
on which no further rewrite rules are applicable.
We refer to $t$ as  the canonical form of $r$.
It is clear that $r \eqsim t$.

\begin{lemma}
Let $r$ and $s$ be behaviors such that $r \eqsim s$.
Then, we find $r \rightarrowtail^* t_1$ and $s \rightarrowtail^* t_2$
for some $t_1$ and $t_2$ where $t_1$ and $t_2$ are syntactically equivalent.
\end{lemma}
\begin{proof}
The result follows by induction over the derivation $r \eqsim s$.
For brevity, we only show some cases.

Suppose
\bda{c}
\rlabel{Compatibility} ~ \myirule{s_1 \eqsim s_2}{r[s_1] \eqsim r[s_2]}
\eda
By the IH, we find $s_1 \rightarrowtail^* t$, $s_2 \rightarrowtail^* t$
for some $t$.
Hence, $r[s_1] \rightarrowtail^n r[t]$
and $r[s_2] \rightarrowtail^m r[t]$.
That is, in finite number of rewrite steps (either $n$ or $m$),
we rewrite $r[s_1]$ into  $r[t]$
 and $r[s_2]$ into $r[t]$.
As every term, e.g.~$r[t]$, must have a canonical normal form,
the result follows by combing the above rewrite steps
and from the fact that the rewrite system is confluent.

Suppose in the last derivation step, we apply
\bda{c}
\rlabel{Associativity} ~ r_1 + (r_2 + r_3) \eqsim (r_1 + r_2) + r_3
\eda
Suppose $r_1 + (r_2 + r_3) \rightarrowtail^* t$.
Then, $(r_1 + r_2) + r_3 \rightarrowtail^{Assoc2} r_1 + (r_2 + r_3)$
and by confluence we know that 
$r_1 + (r_2 + r_3)$ and  $(r_1 + r_2) + r_3$ share the same canonical 
normal form.
\qed
\end{proof}

Based on the above results we 
 obtain a decidable method to test for similarity
for some behaviors $r$ and $s$.
We build the canonical normal forms of $r$ and $s$.
If the canonical normal forms are syntactically equivalent,
we find that that $r$ and $s$ are similar.
Otherwise, $r$ and $s$ are dissimilar.

\subsection{Normal form}
\label{sec:normalform}

Define smart constructors for expressions that implement the similarity rules. We assume a total
ordering $<$ on forkable expressions. We further assume that arguments of constructors are already
in normal form. That means, that the monomials in every sum are sorted in ascending order and that
monomials are bracketed to the right.
\begin{displaymath}
  r \oplus s = 
  \begin{cases}
    r & r = s \vee s=\phi \\
    s & r=\phi \\
    r' \oplus (r'' \oplus s) & r=r'+r'' \\
    (r \oplus s') + s'' & r \ne r'+r'', s=s'+s'', r<\min s'' \\
    s + (r \oplus s'') & r \ne r'+r'', s=s'+s'', r\ge \min s'' \\
    r+s& r\ne r'+r'', s\ne s'+s'', r<s \\
    s+r& r\ne r'+r'', s\ne s'+s'', r>s \\
  \end{cases}
\end{displaymath}
\begin{displaymath}
  r \odot s =
  \begin{cases}
    \phi & r=\phi \vee s=\phi \\
    r & s=\varepsilon \\
    s & r = \varepsilon \\
    r' \odot (r'' \odot s) & r = r'\cdot r'' \\
    s \cdot r & r = \forkEff{r'}, s= \forkEff{s'}, s'<r' \\
    \forkEff{s'} \cdot (r \odot s'') & r = \forkEff{r'}, s= \forkEff{s'} \cdot s'', s'<r' \\
    r \cdot s & r \ne r'\cdot r''
  \end{cases}
\end{displaymath}
\begin{displaymath}
  r^\circledast =
  \begin{cases}
    \varepsilon & r = \phi \vee r = \varepsilon \\
    r^* & \text{otherwise}
  \end{cases}
\end{displaymath}
\begin{displaymath}
  F (r) =
  \begin{cases}
    \phi & r=\phi \\
    \varepsilon & r = \varepsilon \\
    \forkEff{r} & \text{otherwise}
  \end{cases}
\end{displaymath}
This normal form exploits the following lemma, which is straightforward to prove from the semantics.
\begin{lemma}
  $\forkEff r \conc \forkEff s \semeq \forkEff s \conc \forkEff r$, for all $r$ and $s$.
\end{lemma}
\section{Nullable and Emptiness Test}

\begin{definition}[Nullable Test]
\begin{align*}
  \nullable{\phi}& = \false
  &
 \nullable{\varepsilon} & = \true
  &
 \nullable{x} &= \true
\\
\nullable{r + s} &= \nullable{r} \vee \nullable{s}
 &
\nullable{r \conc s}& = \nullable{r} \wedge \nullable{s}
 &
\nullable{r^*} &= \true
\\
\nullable{\forkEff{r}} &= \nullable{r}
\end{align*}
\end{definition}

\begin{lemma}
Let $r$ be a behavior.
Then, $\varepsilon \in L(r)$ iff $\nullable{r}$.
\end{lemma}

\begin{definition}[Emptiness Test]
\label{def:emptiness-test}
\begin{align*}
 \isEmpty{\phi}& = \true
  &
 \isEmpty{\varepsilon}& = \false
  &
 \isEmpty{x}& = \false
\\
\isEmpty{r + s}& = \isEmpty{r} \wedge \isEmpty{s}
 &
\isEmpty{r \conc s}& = \isEmpty{r} \vee \isEmpty{s}
 &
\isEmpty{r^*} & = \false
\\
\isEmpty{\forkEff{r}} & = \isEmpty{r}
\end{align*}
\end{definition}

\begin{lemma}
Let $r$ be a behavior.
Then, $L(r) = \{ \}$ iff $\isEmpty{r}$.
\end{lemma}
\begin{proof}
Fact: $L(r, \{ \}) = \{ \}$.
We verify $L(r,K) \not = \{ \}$ iff $\neg \isEmpty{r}$.
The proof proceeds by straightforward induction.
\qed
\end{proof}

\end{document}